\documentclass[a4paper, onecolumn, accepted=2023-08-22]{quantumarticle}
\pdfoutput=1
\usepackage[utf8]{inputenc}
\usepackage[english]{babel}
\usepackage[T1]{fontenc}
\usepackage{amsmath}
\usepackage{amssymb}
\usepackage{amsthm}
\usepackage{amsfonts}
\usepackage{hyperref}
\usepackage{subcaption}
\usepackage{tikz}
\usepackage{lipsum}
\usepackage{color}
\usepackage{cleveref}
\usepackage{enumerate}
\usepackage{array}
\usepackage{xcolor}
\usepackage{braket}
\usepackage{algpseudocode}
\usepackage{algorithm}

\newcommand{\ff}{\mathbb{F}_2}
\newcommand{\hgp}{\text{HGP}}
\DeclareMathOperator{\rank}{\mathrm{rank}}
\DeclareMathOperator{\Span}{\mathrm{Span}}

\DeclareMathOperator{\supp}{\mathrm{supp}}
\newcommand{\ox}{\otimes}
\newcommand{\syhgp}{\mathrm{HGP_{sy}}}
\newcommand{\bline}{\mathcal{B}_{\mathrm{line}}} 
\DeclareMathOperator{\im}{\mathrm{Im}}
\newtheorem{theorem}{Theorem}
\newtheorem{definition}{Definition}
\newtheorem{lemma}{Lemma}
\newcommand{\gates}[1]{\texttt{#1}}  

\newcolumntype{L}{>{$}l<{$}} 
\newcolumntype{C}{>{$}c<{$}} 
\begin{document}

\title{Partitioning qubits in hypergraph product codes to implement logical gates}

\author{Armanda O. Quintavalle}
\affiliation{Department of Physics \& Astronomy, University of Sheffield, Sheffield, S3 7RH, United Kingdom}
\affiliation{Dahlem Center for Complex Quantum Systems, Freie Universität Berlin, 14195 Berlin, Germany}
\author{Paul Webster}
\affiliation{Centre for Engineered Quantum Systems, School of Physics, The University of Sydney, Sydney, NSW 2006, Australia}
\author{Michael Vasmer}
\affiliation{Perimeter Institute for Theoretical Physics, Waterloo, ON N2L 2Y5, Canada}
\affiliation{Institute for Quantum Computing, University of Waterloo, Waterloo, ON N2L 3G1, Canada}

\affiliation{}
\orcid{}
\begin{abstract}
The promise of high-rate low-density parity check (LDPC) codes to substantially reduce the overhead of fault-tolerant quantum computation depends on constructing efficient, fault-tolerant implementations of logical gates on such codes.
Transversal gates are the simplest type of fault-tolerant gate, but the potential of transversal gates on LDPC codes has hitherto been largely neglected. 
We investigate the transversal gates that can be implemented in hypergraph product codes, a class of LDPC codes.
Our analysis is aided by the construction of a symplectic canonical basis for the logical operators of hypergraph product codes, a result that may be of independent interest.
We show that in these codes transversal gates can implement Hadamard (up to logical SWAP gates) and control-Z on all logical qubits.
Moreover, we show that sequences of transversal operations, interleaved with error-correction, allow implementation of entangling gates between arbitrary pairs of logical qubits in the same code block. 
We thereby demonstrate that transversal gates can be used as the basis for universal quantum computing on LDPC codes, when supplemented with state injection.
\end{abstract}

\section{Introduction}
In recent years, quantum computing has transitioned from a theoretical idea to a real technology that is competitive with state-of-the-art classical computing on carefully selected tasks~\cite{Arute}. 
However, realising its potential to solve intractable problems of practical importance requires the development of a fault-tolerant quantum computer---a device that can perform long quantum computations to a very high degree of accuracy even in the presence of noise~\cite{shor1997}. 
Fault-tolerant quantum computing can be achieved by encoding quantum information into quantum error-correcting codes, but only at the cost of substantial overhead. 
For standard fault-tolerant quantum architectures based on the surface code this overhead is prohibitively large for the foreseeable future, with approximately $10^6$ to $10^8$ qubits being required to realise useful applications~\cite{gidney2021factor,Kim}. 

High-rate quantum LDPC codes offer a promising avenue to significantly reducing this overhead~\cite{Breuckmann}. 
These codes move beyond topological codes such as the surface code~\cite{Kitaev} by substituting the requirement that stabiliser check operators are geometrically local with a weaker condition that they must be sparse. They thereby preserve the essential benefits of the surface code for error-correction, while allowing much smaller qubit overheads~\cite{Breuckmann, panteleev2022asymptotically}.

In this work, we focus on hypergraph product codes---a class of high-rate LDPC codes derived from the tensor product of classical LDPC codes. 
Even when restricted to this highly-structured class of codes, it is not yet clear what is the best strategy to reliably perform logic. 
The earliest proposal for fault-tolerant logic on LDPC codes relies on state injection protocols that, while achieving a constant overhead per logical qubit in the asymptotic regime, would have substantial finite-size overheads~\cite{Gottesman}. In~\cite{Krishna, krishna2020topological}, code deformation on hypergraph product codes is used to perform Clifford gates. 
The protocol proposed preserves the LDPC property of the code throughout the entire computation; however, even if in principle the whole Clifford group is implementable via this framework, there is no promise that this is the case for an arbitrary hypergraph product code nor is there a promise on the time cost of each gate implementation. 
Another variation of code deformation, via non-destructive measurement of high-weight logical operators, is proposed in~\cite{Cohen}; the method there proposed enables the implementation of the full Clifford group on all hypergraph product codes. 
Nevertheless, this method requires many ancilla qubits if we want to operate in parallel on all of the encoded qubits.

We take a different approach to fault-tolerance and explore transversal gates~\cite{zeng,Eastin,bravyikoenig,pastawski,jochym2018disjointness,Webster}.
Transversal gates on homological codes~\cite{bravyi14}, close cousins of the hypergraph product codes, were studied in~\cite{Jochym-O'Connor2}.
Universality is there obtained by combining two quantum codes with complementary transversal gates. 
Hence in \cite{Jochym-O'Connor2} the problem of universality is deferred to the one of finding complementary classes of good quantum codes that are then combined via the homological product; the information is always protected by at least one code whilst leveraging the transversal gates of the other. We relax the constraints on the underlying (classical) codes used as seed codes in the hypergraph product and investigate how the core symmetries that arise from the product structure itself enable transversal gates.

The transversal gates we find are only a small subset of the Clifford group. 
Nonetheless, we further the knowledge on hypergraph product codes by proving that it is possible to efficiently construct a canonical basis for their logical space.  The existence of such a canonical basis is non-trivial and may be of independent interest in the development of other computational frameworks on hypergraph product codes or related families of codes (e.g.\ higher dimensional homological product codes~\cite{bravyi14, Quintavalle2, jochym2021four}, or the codes introduced in~\cite{evra2022decodable}).
Furthermore, we illustrate how pieceably fault-tolerant circuits~\cite{Yoder} and state injection can be used on symmetric hypergraph product codes \cite{kovalev2012improved} in conjunction with the transversal gates proposed to give a universal gate set. 
Importantly, our scheme is relevant to small, low-overhead quantum error-correcting codes with practical near-term potential.

After a short introduction to hypergraph product codes in \Cref{sec:hgp}, we detail the existence of the canonical basis for their logical space in \Cref{sec:canonical_basis}.
In \Cref{sec:transversal_gates}, we illustrate how the idea of unfolding the color code into surface codes relies on more general symmetries of the product structure, symmetries that are inherited by square and symmetric codes.
Thanks to this observation and the use of our canonical basis, we showcase the transversal implementation of some Clifford gates on hypergraph product codes. 
We comment on how pieceable fault tolerance and state injection, in combination with the transversal gates introduced, could be used on `small' codes to perform arbitrary computations in \Cref{sec:universal}.

\section{Hypergraph Product Codes}
\label{sec:hgp}
Any binary matrix $H$ in $\ff^{m \times n}$ defines a classical code whose codewords are vectors in the kernel of $H$, $\ker H$. 
The classical code so defined uses $n$ bits to protect from errors $k = n - \rank(H)$ bits and detects all errors of Hamming weight less than $d$, the minimum weight of a codeword. 
Briefly, we say that $H$ defines an $[n, k, d]$ code~\cite{Mackay}. 
Similarly, any pair of binary matrices $H_x$ in $\ff^{m_x \times n}$ and $H_z$ in $\ff^{m_z \times n}$ such that $H_x \cdot H_z^T = 0 \mod 2$ defines a stabiliser CSS code~\cite{calderbank1996good, steane1996multiple, nielsen2002quantum}. 
The stabiliser group is $\mathbf{S}=\langle \mathbf{S_x} \cup \mathbf{S_z}\rangle$, where $\mathbf{S_x}$ is the set of $X$ operators whose support vector\footnote{Given a Pauli operator on $n$ qubits $P = P_1 \otimes P_2 \otimes \dots \otimes P_n \in \mathcal{P}_n$, its support vector is the unique vector $v \in \ff^n$ with $i$th coordinate $v[i]= 1$ if and only if $P_i \neq I$.} is equal to a row of $H_x$, and $\mathbf{S}_z$ is the set of $Z$ operators whose support vector is equal to a row of $H_z$. 
Since an $X$ and a $Z$ operator commute if and only if their supports have even overlap, the orthogonality of $H_x$ and $H_z$ in $\ff$ ensures that the stabiliser group $\mathbf{S}$ is well defined. 
The quantum code with stabiliser group $\mathbf{S}$ uses $n$ physical qubits to encode $k = n - \rank(H_x) - \rank(H_z)$ logical qubits and detects all errors of weight less than $d = \min(d_x, d_z)$, where $d_x$ is the minimum weight of a vector in $\ker H_z$ that is not in the image of $H_x^T$, $\im H_x^T$, and $d_z$ is the minimum weight of a vector in $\ker H_x$ that is not in $\im H_z^T$. 
Briefly, we say that the quantum code is $[\![n, k, d]\!]$.

A hypergraph product code~\cite{tillich14, bravyi14, audoux2015tensor} is a CSS code produced from the hypergraph product of two linear codes. 
Given $H_a$ in $\ff^{m_a \times n_a}$ and $H_b$ in $\ff^{m_b \times n_b}$ we define:
\begin{align}
H_x&=\left(H_a\otimes I_{n_b} \quad I_{m_a}\otimes H_b^T\right),\label{eq:h_x} \\
H_z &=\left( I_{n_a}\otimes H_b  \quad H_a^T\otimes I_{m_b}\right),
\label{eq:h_z}
\end{align}
where $\otimes$ is the tensor product.
Since $H_x \cdot H_z^T = 2 H_a \ox H_b^T = 0 \mod 2$, this choice is valid and we indicate by $\hgp(H_a, H_b)$ this CSS code. 
If $H_{\eta}$ (resp. $H_{\eta}^T$) define a $[n_{\eta}, k_{\eta}, d_{\eta}]$ (resp. $[m_{\eta}, k_{\eta}^T, d_{\eta}^T]$) classical code, then $\hgp(H_a, H_b)$ has parameters 
\begin{align}
\label{eq:parameters}
[\![n_an_b+ m_am_b, k_ak_b + k_b^Tk_a^T, d]\!],
\end{align}
where $d = \min(d_a, d_a^T, d_b, d_b^T)$.

The physical qubits of $\hgp(H_a, H_b)$ can be arranged in two rectangular grids of sizes $n_a\times n_b$ and $m_a \times m_b$, see \cref{fig:logical_qubits_enumeration} and~\cite{quintavalle2022reshape}. 
As such, we enumerate the physical qubits of $\hgp(H_a, H_b)$ by the triplet $(i,h,L)$ and $(j, \ell, R)$, where $1\le i \le n_a$, $1\le h \le n_b$, $1 \le j \le m_a$, $1 \le \ell \le m_b$. 
The first two coordinates of each triplet refer to the row and column positions of the physical qubit in the grid. 
The third coordinate $L$ or $R$, short for left and right respectively, distinguishes the two grids and is referred to as \emph{sector} of the physical qubit. 
Via this spacial mapping of physical qubits, each stabiliser generator will have support contained in a row of one sector and a column of the other. 
Specifically, elements of $\mathbf{S_x}$ can be indexed by a pair $j,h$, such that $S_{x}(j,h)$ has support on qubits on row $j$ of the right sector and column $h$ of the left sector. 
Elements of $\mathbf{S_z}$, indexed by a pair $i, \ell$ are similar, but have support on rows of qubits in the left sector and columns in the right sector.

We say that a hypergraph product code is a \emph{square} code if it is derived from one classical matrix only. 
Square codes $\hgp(H, H)$ take their name from the shape of the two grids of physical qubits: if $H \in \ff^{m \times n}$, the left and right sectors are square grids of sizes $n \times n$ and $m \times m$ respectively. 
A square code $\hgp(H, H)$ such that $H = H ^T$ or, more loosely, such that $H \in \ff^{n \times n}$ is square and $\im(H) = \im(H^T)$, is said \emph{symmetric} and denoted as $\syhgp(H)$.

\subsection{Logical Pauli Operators}
\label{sec:canonical_basis}
Logical $X$ and $Z$ operators of $\hgp(H_a, H_b)$ can be chosen to have support on a `line' of qubits, meaning that each operator acts non-trivially either on the left or right sector, and either on a column or on a row of physical qubits~\cite{kovalev2013quantum, bravyi14, quintavalle2022reshape}. 
We build on this and show that such a `line basis' can, in addition, be chosen to be symplectic. 
The existence of such canonical basis is the key ingredient in the identification of transversal gates on hypergraph product codes. 
Formally:
\begin{theorem}
\label{thm:bline}
For any hypergraph product code there exist bases $\bline^x$ and $\bline^z$ of logical $X$ and $Z$ operators respectively such that:
\begin{enumerate}
    \item Any operator in $\bline^x$ or $\bline^z$ has support on a `line' of qubits; by line here we intend that the support of the operator is contained in either a column or a row of left or right sector qubits, when qubits are arranged on two grids as explained above.
    \item For any operator in $\bline^x$ there exists one and only one operator in $\bline^z$ that anticommutes with it. More precisely for every pairs of operators, one in $\bline^x$ and one in $\bline^z$, either their support overlaps on exactly one qubit or does not overlap at all.
\end{enumerate}
We refer to any such pair of bases as canonical basis for $\hgp(H_a, H_b)$.
\end{theorem}

The proof of \Cref{thm:bline} is constructive and relies on a modified version of the Gaussian elimination algorithm over $\ff$ (\Cref{algo:strong_triangular}), which yields a special kind of triangular matrices that we call strongly lower triangular matrices. 
\begin{definition}
\label{def:strong_triang}
An $m \times n$ matrix $A$ is said to be strongly lower triangular if:
\begin{enumerate}
    \item Any column $j$ has a pivot $p$ such that all the elements below the pivot are zero.
    \item All the pivots are distinct.
    \item Reordering if necessary, for any pivot $A_{p, j}=1$, the coefficients to its right are zero i.e.\ $A_{p, j+1}= A_{p, j+2}= \dots A_{p, m} = 0$.
\end{enumerate}
We indicate by $\pi(A)$ the set of row pivots of $A$:
\begin{align*}
  \pi(A) = \{p \text{ s.t. } &A_{p, j} \text{ is a pivot for some column index } j \text{ of A}\} 
\end{align*}
Given a vector space, we say that a set of vectors is strongly lower triangular if its matrix representation is.
\end{definition}
An example of a strongly lower triangular matrix is:
\begin{align}
\label{eq:a_matrix}
A = \begin{pmatrix}
        1 & 1 & 0 & 1\\
        1 & 1 & 1 & 0\\
        1 & 0 & 0 & 0\\
        0 & 1 & 1 & 1\\
        0 & 1 & 0 & 0\\
        0 & 0 & 1 & 0\\
        0 & 0 & 0 & 1
\end{pmatrix}
\end{align}
where $\pi(A) = \{3, 5, 6, 7\}$.
Observe in particular the pivot $A_{3, 1}$: all the coefficients to its right are zero even if the submatrix on its right is not zero.

As \Cref{lemma:basis} below shows, strongly lower triangular bases for the classical codes and their transposes used as seed codes in the hypergraph product naturally give a canonical basis for the associated hypergraph product code. 
\begin{lemma}
\label{lemma:basis}
Given strongly lower triangular bases for $\ker H_a$, $\ker H_a^T$, $\ker H_b$ and $\ker H_b^T$, it is possible to construct a canonical basis for the associate hypergraph product code $\hgp(H_a, H_b)$.
\end{lemma}
\begin{proof}
We begin with some notation. Given a vector space $V \subseteq \ff^n$ its complement is any vector space $V^{\bullet} \subseteq \ff^n$ such that $V \oplus V^{\bullet} = \ff^n$. Crucially, $V^{\bullet}$ has same dimension as the orthogonal complement $V^{\bot}$ of $V$:
\begin{align*}
    V^{\bot} = \{ w \in \ff^n \text{ s.t. } \langle v, w\rangle = 0 \text{ for all } v \in V\},
\end{align*}
but they are, in general, different spaces. For example, if $V = \Span \left( (1, 1, 1)^T, (0, 1, 0)^T\right) $ then $V^{\bot} = \Span \left( (1, 0, 1)^T\right) $ and a complement of $V$ is  $V^{\bullet} = \Span \left((0, 0, 1)^T \right)$.

Before detailing the proof of \Cref{lemma:basis}, we remind the reader that the logical $Z$ operators of $\hgp(H_a, H_b)$ are spanned by:
\begin{align*}
    \left\{ (k \otimes \bar f, 0) \text{ for } k \in \ker H_a, \bar f \in (\im H_b^T)^{\bullet}\right\} \cup \left\{ (0, g \otimes \bar h) \text{ for } \bar h \in \ker H_b^T, g \in (\im H_a)^{\bullet}\right\},
\end{align*}
and the logical $X$ operators are spanned by:
\begin{align*}
    \left\{ ( \bar f' \otimes k', 0) \text{ for } k' \in \ker H_b, \bar f' \in (\im H_a^T)^{\bullet}\right\} \cup \left\{ (0, \bar h' \otimes g') \text{ for } \bar h' \in \ker H_a^T, g' \in (\im H_b)^{\bullet}\right\},
\end{align*}
where for an $m \times n$ matrix $H$, $\im H \subseteq \ff^m$ is the space spanned by the columns of $H$ and $\im H^T \subseteq \ff^n$ is the space spanned by the rows of $H$. See, for instance,~\cite{bravyi14, quintavalle2022reshape}. Thus, the problem of finding a canonical basis for $\hgp(H_a, H_b)$ can be reduced to the problem of finding suitable bases for the kernels and the complement spaces defined by the classical seed matrices. In the proof below we follow this approach and show that strongly lower triangular bases for the classical seed matrices indeed induce a canonical basis for $\hgp(H_a, H_b)$.

Let $A = \{a_{i}\} \in \ff^{n_a}$, $\bar A = \{\alpha_j\}\in \ff^{m_a}$, $B = \{b_h\} \in \ff^{n_b}$ and $\bar B = \{\beta _{\ell}\} \in \ff^{m_a}$ be strongly lower triangular bases for  $\ker H_a$, $\ker H_a^T$, $\ker H_b$ and $\ker H_b^T$ respectively. 
The indices are the pivots: $a_i$ has $i$th coordinate $a_i[i] =1$ and for any $i' > i$, its $i'$th coordinate $a_i[i'] = 0$.

We claim that if $f_i \in \ff^{n_a}$ is the $i$th unit vector, then the set $A_p =\{ f_i\}_{i \in \pi(A)} \subseteq \ff^{n_a}$ is a basis of a complement $(\im H_a^T)^{\bullet}$ of $\im H_a^T$. Since the size $n_a - \rank(H_a^T)$ of $A_p$ equals the dimension of $(\im H_a^T)^{\bullet}$ and the vectors in $A_p$ are clearly linearly independent, we only need to prove that they generate a space that has trivial intersection with $\im H_a^T$ i.e. $\Span \left(A_p \right)\cap \im H_a^T = \{0\}$. The following, is a proof by contradiction. First, we assume that there exists a vector in the set $A_p$ that belongs to $\im H_a^T$ and we reach a contradiction. In particular, this entails that $A_p \cap \im H_a^T = \emptyset$. Second, in order to obtain a contradiction, we assume that there exists a linear combination of vectors in $A_p$ that belongs to $\im H_a^T$. We reach a contradiction using the definition of the vectors $f_i$ in $A_p$ and the linearity of the inner product.

Let assume that $f_i \in \im H_a^T$ for some $f_i \in A_p$. Then, $f_i$ can be written as a linear combination of a set $\rho$ of rows of $H_a^T$:
\begin{align}
\label{eq:sum_image}
    \sum_{v^T \in \rho } v = f_i .
\end{align}
By definition of strongly lower triangular basis $A = \{a_i\}$ of $\ker H_a$, row pivots $\pi(A)$ and unit vectors $A_p = \{f_i\}$, for any $i, i' \in \pi(A)$, $\langle f_i, a_{i'} \rangle = 1$ if and only if $i'=i$ and $0$ otherwise. Combining this with \cref{eq:sum_image}, yields:
\begin{align*}
    1 &= \langle f_i, a_{i}, \rangle &\text{ }\\
    &= \left\langle\sum_{v^T \in \rho} v,\, a_i \right\rangle &\text{by hypothesis,}\\
    &= \sum_{v^T \in \rho}\langle v, a_i\rangle&\text{by linearity,}\\
    &= 0 &v \in \im H_a^T, a_i \in \ker H_a ,
\end{align*}
which is a contradiction. In particular, see \cref{eq:sum_image}, no vector $f_i$ in the set $A_p$ belongs to $\im H_a^T$: $A_p \cap \im H_a^T = \emptyset$ and therefore $A_p \subseteq (\im H_a^T)^{\bullet}$. Moreover, if we assume that some linear combination $f$ of vectors in $A_p$ belongs to $\im H_a^T$, again we could write $f$ as a sum of vectors in $\im H_a^T$ as in \cref{eq:sum_image}, and via the same argument we used for a unit vector $f_i$ thanks to the linearity of the inner product, we find a contradiction. In other words, no non-trivial linear combination of vectors in $A_p$ belongs to $\im H_a^T$: $\Span(A_p) \cap \im H_a^T = \{0\}$ and so $\Span(A_p) \subseteq (\im H_a^T)^{\bullet}$ by definition of complement space. Via a dimension-counting argument, since $\dim (A_p)= \dim(\im H_a^T)^\bullet$ and vectors in $A_p$ are linearly independent, we conclude that $A_p$ is a basis of $(\im H_a^T)^{\bullet}$.

Similarly, we can prove that $\bar A_p$, $B_p$ and $\bar B_p$ are well defined bases of the corresponding complement spaces.

Because the tensor product of bases is a basis of the tensor product space, the sets:
\begin{align}
\label{eq:bline_x}
    \bline^x &= \{(f_i \ox b_h, 0), (0, \alpha_{j} \ox f_{\ell})\},\\
    \bline^z &= \{(a_i \ox f_h, 0), (0, f_{j} \ox \beta_{\ell})\},
\label{eq:bline_z}
\end{align}
for $i \in \pi(A)$, $h \in \pi(B)$, $j \in \pi (\bar A)$, and $\ell \in \pi (\bar B)$, are sets of linearly independent logical operators. Moreover, it is straightforward to see that operators in $\bline^x$ and $\bline^z$ have support on a line of qubits. Thus, in order to verify that $\bline = \bline^x \cup \bline^z$ is a canonical basis for $\hgp(H_a, H_b)$, we only need to show that it is symplectic. 

Clearly, left and right operators do not overlap. 
If instead we take a left $Z$ operator and a left $X$ operator we find:
\begin{align}
\label{eq:two_logicals}
  \langle (a_i \ox f_h, 0), \, (f_{i'} \ox b_{h'}, 0) \rangle &= a_i[i'] \cdot b_{h'}[h], 
\end{align}  
and by strong lower triangularity:
\begin{align*}
a_i[i'] \cdot b_{h'}[h] & =\begin{cases} 1, \text{ if } i= i' \text{ and } h = h',\\
  0, \text{ otherwise }.
  \end{cases}
\end{align*}
Since an equivalent relation holds for pairs of right operators, we have shown that for every operator in $\bline^x$ there exists a unique operator in $\bline^z$ that is not orthogonal to it. Furthermore, two overlapping operators of different type overlap on exactly one physical qubit e.g.\ for the two left operators in \cref{eq:two_logicals}, it is the physical qubit at position $(i, h, L)$ when $i = i'$ and $h = h'$. 

In conclusion, $\bline = \bline^x \cup \bline^z$ is a canonical basis for $\text{HGP}(H_a, H_b)$ as per \Cref{thm:bline}.
\end{proof}
\Cref{thm:bline} is a corollary of  \Cref{lemma:basis}, provided that strongly lower triangular bases exist and can be found. We show how this is in fact the case in \Cref{app:proof}, where we present a modification of the Gaussian reduction algorithm, \Cref{algo:strong_triangular}, which finds a strongly triangular basis for any $n \times n$ binary matrix in time $O(n^3)$.

Practically, if an arbitrary code has a canonical basis as per \Cref{thm:bline}, any indexing of the physical qubits of the code naturally yields an indexing of the logical qubits. Namely, if the unique physical qubit in the overlap of a basis logical $X$ operator and a basis logical $Z$ operator has index $\iota$, then the index of the corresponding logical qubit is $\iota$ too, and we indicate it as $q_\iota$. We call the physical qubit $\iota$ the physical \emph{pivot} of the logical qubit $q_{\iota}$.

For a hypergraph product code $\hgp(H_a, H_b)$, if we display qubits on a grid as explained above and we fix a canonical basis as per \Cref{thm:bline}, we can uniquely index logical qubits as $q_{i,h}^L$ and $q_{j, \ell}^R$ corresponding to the pivots $(i, h, L)$ and $(j, \ell, R)$. 
By construction, the sector $L$ or $R$ indicates whether the logical qubit has canonical operators supported on the left or right sector qubits and the first two coordinates specify where the two canonical operators cross. 
We refer to physical qubits $(i, i, \sigma)$, $\sigma = L, R$, as \emph{diagonal} qubits and to the others as \emph{mirror} qubits. 
The logical qubits inherit the same attribute of their pivots. 
To sum up, we have found a classification of the logical qubits of $\hgp(H_a, H_b)$ that depends on the physical position of the crossing of the logical $X$ and $Z$ operators, $\bar X$ and $\bar Z$, see \cref{tab:notation_qubits}. 

\begin{table}[htb]
\centering
\begin{tabular}{C|C|C|C}
\text{symbol} & \bar Z & \bar X & \text{ pivot qubit}\\
\hline
q^{L}_{i, h} & (a_i \ox f_h, 0) & (f_i \ox b_h, 0) & (i, h, L)\\
q^{R}_{j, \ell} & (0, f_j \ox \beta_{\ell}) & (0,  \alpha_j \ox f_{\ell}) & (j, \ell, R)
\end{tabular}
\caption{Classification and indexing of the logical qubits of a hypergraph product code $\text{HGP}(H_a, H_b)$ where $\{a_i\}, \{b_h\}, \{\alpha_j\}, \{\beta_\ell\}$ are strongly lower triangular bases of $\ker H_a, \ker H_b, \ker H_a^T, \ker H_b ^ T$ indexed over the basis' pivots. We refer to logical qubits whose pivot lies on the diagonal e.g.\ $q_{i,i}^L$ or $q^R_{j, j}$ as diagonal logical qubits, and to the others as mirror logical qubits. 
\label{tab:notation_qubits}}
\end{table}

\section{Transversal Clifford gates}
\label{sec:transversal_gates}
Transversal logical operators offer the most straightforward approach to realising fault-tolerant quantum computation since they naturally limit the spread of errors. 
A unitary operator $U$ is transversal with respect to a partition\footnote{A partition of a set is a collection of non-empty and disjoint subsets of the set, whose union is the whole set.} $\{Q_i\}$ of the physical qubits of an $[\![n, k, d]\!]$ code if it can be expressed as $U=\bigotimes_i U_i$, where each $U_i$ acts non-trivially only on qubits in $Q_i$~\cite{jochym2018disjointness}. The usual notion of transversal gates~\cite{Eastin} is found choosing the singleton partition $\left\{ \{i\}\right\}$. By construction transversal gates do not spread errors between qubits in different subsets and are therefore inherently fault-tolerant, provided that all the subsets in the partition are correctable\footnote{A set of qubits is said to be correctable if it cannot contain the support of a non-trivial logical operator.}. We say that a partition $
\{Q_i\}$ is $m$-local if all its subsets have size at most $m$: $\lvert Q_i \rvert \le m$. The partition-distance $\delta_{\{Q_i\}}$ is the minimum number of subsets in the partition that supports a logical operator. Equivalently, the code can detect all errors that are supported on at most $\delta_{\{Q_i\}} - 1$ subsets. The partition-distance is a measure of how many faulty factors $U_i$ the code can tolerate without corrupting the logical information. In the same way, we can think of the distance of a code as a measure of how many faulty `identity factors' a code can deal with. 
As an example, for any $[\![n, k, d]\!]$ stabiliser code, the singleton partition $\left\{\{i\}\right\}$ is $1$-local, has partition distance $\delta_{\left\{\{i\}\right\}} = d$ and the logical Pauli operators are transversal with respect to it.

In this Section we propose a partition for square codes of partition-distance $\lfloor d/2 \rfloor$ and one for symmetric codes of partition-distance $d$; for both partitions, we report some examples of transversal operators.
Importantly, as suggested in~\cite{jochym2018disjointness, Burton}, we expect transversal operators on (2-dimensional) hypergraph product codes to be restricted to be either Pauli or Clifford operators, hence we here focus on Clifford operators.

Clifford operators permute the Pauli operators by mapping the Pauli group on $n$ qubits $\mathcal P_n$ into itself. 
The set of Clifford gates on $n$ qubits, $\mathcal{C}_n$, is a group, that can be generated by:
\begin{enumerate}[(i)]
\item The Hadamard gate H, that maps $X$ operators into $Z$ operators and vice-versa,  $X \leftrightarrow Z$.
\item The phase gate S, that maps $X$ operators into $Y$ operators and fixes $Z$ operators, $X \rightarrow Y$.
\item The CZ gate, a two qubit gate that maps $X \otimes I \rightarrow X \otimes Z$, $I \otimes X \rightarrow Z \otimes X$ and acts trivially on $Z$ operators.
\end{enumerate}

\subsection{Gates on square codes}
\label{sec:transversal_sq}
The first partition we propose builds on the unfolding technique for the color code proposed in~\cite{kubica2015unfolding} and further studied in~\cite{Moussa2016,vasmer2019three, breuckmann2022fold}.
In~\cite{kubica2015unfolding, Moussa2016}, the equivalence between color codes and folded planar codes is leveraged to construct transversal Clifford gates on the planar code. Here we build on that same idea to investigate Clifford gates on square hypergraph product codes. 

Hypergraph product codes can be seen as a generalisation of the planar code, which indeed is the hypergraph product code $\hgp(H_{\mathrm{rep}}, H_{\mathrm{rep}})$, where $H_{\mathrm{rep}}$ is the full-rank parity check matrix of the repetition code, e.g.\
\begin{align}
    H_{\mathrm{rep}} = \begin{pmatrix}
    1 & 1 & 0\\
    0 & 1 & 1
    \end{pmatrix},
\end{align}
for parameters $[3, 1, 3]$. By exploiting the symmetries of the canonical basis of hypergraph product codes  (\Cref{thm:bline}), we are able to generalise the folding of~\cite{kubica2015unfolding} to all square hypergraph product codes $\hgp(H, H)$. We can fold the left and right grid of qubits along the principal diagonal and pair the physical qubits whose sites overlap upon folding, see \cref{fig:folding}. 
Upon folding, mirror physical qubits are twinned:  $(i, h, L)$ twins with $(h, i, L)$ and $(j, \ell, R)$ with $(\ell, j, R)$. 
We call the partition given by singletons of diagonal qubits and two-qubit sets of twin qubits, diagonal-twin partition. 
More precisely, if $H \in \ff^{m \times n}$ and $\hgp(H, H)$ is a $[\![n, k, d]\!]$ code, the diagonal-twin partition of its physical qubits is given by
\begin{align}
\label{partition_twin_diagonal}
&\Big\{\{(i, i, L)\}, \{(j, j, R)\}\Big\} \notag\\
&\hspace{1.7cm}\cup\\
&\Big\{ \{(i, h, L),\, (h, i, L)\}, \{(j, \ell, R),\, (\ell, j, R)\} \Big\} \notag
\end{align}
where $1\le i, h \le n$ and $i \neq h$; $1 \le j, \ell \le m$ and $j \neq \ell$.

The diagonal-twin partition is $2$-local and has partition-distance $\delta_{\mathrm{dt}} = \lfloor d/2 \rfloor$. In fact, as proven in~\cite{quintavalle2022reshape}, a non-trivial logical operator for an $[\![n, k, d]\!]$ hypergraph product code has support on at least $d$ rows or columns of qubits in the same sector. Because the diagonal-twin partition is $2$-local, the union of any choice of $\mu$ subsets from it has size at most $2\mu$. Thus, in order to fill at least $d$ rows or $d$ columns of physical qubits, we need to pick at least $\mu\ge  d/2$ subsets in the diagonal-twin partition.

The nomenclature of physical qubits as diagonal and twins is naturally inherited by the logical qubits via the correspondence of logical qubits and their physical pivots, see \cref{tab:notation_qubits}. 
This labelling is key in understanding the logical actions of transversal operations for the diagonal-twin partition and therefore we summarize it here: diagonal logical qubits are indexed as $q_{i,i}^L$ and $q_{j, j}^R$; twin logical qubits as $(q_{i, h}^L, q_{h, i}^L)$ and $(q_{j, \ell}^R, q_{\ell, j}^R)$ for $i \neq h$ and $j \neq \ell$, see \cref{fig:logical_qubits_enumeration}.
\begin{figure}
\centering\includegraphics[scale=0.45]{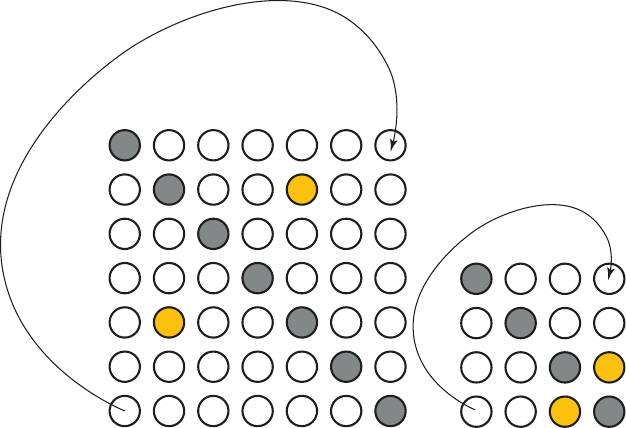}
\caption{Graphical representation of the diagonal-twin partition for squares hypergraph product codes derived from the color code unfolding into the planar code~\cite{kubica2015unfolding, Moussa2016}. The two grids of physical qubits are folded separately along the principal diagonal and qubits that share the same sites upon folding are paired together. The grey circles represent physical diagonal qubits while all the other circles represent mirror qubits. The yellow circles (two on the left, and two on the right) are twin qubits.}
\label{fig:folding}
\end{figure}

Similarly to the planar code case, both the \gates{Hadamard-SWAP} and the $\gates{CZ-S}$ gates detailed below are valid logical operators on $\hgp(H, H)$. Whilst it is immediate to verify that the stabilisers group is preserved by these two operators, less immediate is the identification of the logical operation performed. 
Crucially, this task becomes trivial when we look at the action induced on a canonical basis derived by strongly lower triangular reduction as per \Cref{lemma:basis}. 

On the physical level, \gates{Hadamard-SWAP} consist of (i) Hadamard on every physical qubit; (ii) physical SWAP between twin qubits. \gates{Hadamard-SWAP} is a valid logical operator as it preserves the stabiliser group mapping $S_{x}(j,h) \leftrightarrow S_{z}(h,j)$ and (i) swaps the logical $X$ and the logical $Z$ operators of twin qubits: on the left $q_{i, h}^L$ and $q_{h, i}^L$, on the right $q_{j, \ell}^R$ and $q_{\ell, j}^R$; (ii) acts as logical Hadamard on the diagonal qubits.

The operator \gates{CZ-S} is defined on the physical level as: (i) S gate on left diagonal qubits; (ii) $\text{S}^{\dag}$ on right diagonal qubits; (iii) CZ between twin qubits. \gates{CZ-S} preserves the $Z$ stabilisers and maps the $X$ stabiliser $S_x(j, h)$ into $S_x(j, h) S_z(h,j)$. Importantly, the phase factor in the product $S_x(j, h) S_z(h,j)$ is correctly preserved since, by construction, an $X$-stabiliser $S_x(j, h)$ has a diagonal left qubit $(h,h)$ in its support if and only if $H_a[j, h] = 1$, if an only if it has a diagonal right qubit $(j, j)$ in its support too. As such, if an S gate is applied, an $\text{S}^{\dag}$ gate is applied too and the global phase cancels out. Again by looking at the action of the operator $\gates{CZ-S}$ on a canonical basis, since twin logical operators have support on mirror qubits only and each diagonal logical operator has support on exactly one diagonal qubit, we find that the logical action of the operator \gates{CZ-S} is (i) S on left diagonal qubits; (ii) $\text{S}^{\dag}$ on right diagonal qubits; (iii) CZ between twin qubits.

\subsubsection{An example of a square code}
\label{sec:example_h_tilde}
As a guiding example, we consider the square code $\hgp(\tilde H, \tilde H)$, where $\tilde H$ a is non-full-rank parity check matrix of the classical $[7, 4, 3]$ Hamming code:
\begin{align}
\tilde H = \begin{pmatrix}
1& 1& 0& 1& 1& 0& 0\\
1& 0& 1& 1& 0& 1& 0\\
0& 1& 1& 1& 0& 0& 1\\
1& 0& 1& 0& 1& 0& 1
\end{pmatrix}.
\label{eq:tilde_h}
\end{align}
The matrix $\tilde H$ defines an equivalent $[7, 4, 3]$ classical code and $\tilde H^T$ defines a $[4, 1, 3]$ code. By \cref{eq:parameters}, $\hgp(\tilde H, \tilde H)$ is a $[\![65, 17, 3]\!]$ CSS code. 
The columns of $A$ in \cref{eq:a_matrix} are a strongly triangular basis of $\ker \tilde H $ and $\tilde v = \begin{pmatrix}
1 & 0 & 1 & 1 
\end{pmatrix}^T$ generates $\ker \tilde H^T$. 
Trivially, $\{\tilde v\}$ is a strongly lower triangular basis of $\ker \tilde H ^T$, with pivot $\pi(\tilde{v}) = 4$. 
In \cref{fig:logical_qubits_enumeration} qubits are represented as circles. 
We note how these are divided into a $7 \times 7$ grid of left physical qubits and a $4 \times 4$ grid of right physical qubits. 
The characteristic square shape of these two grids makes $\hgp(\tilde H, \tilde H)$ a square code. 
The physical diagonal qubits are the ones that lie across the principal diagonals of the squares (the two black lines in \cref{fig:logical_qubits_enumeration}). 
The black circles correspond to physical pivots, one for each logical qubit of the code. 
As for any other square hypergraph product code derived from a classical $[n, k, d]$ code whose transpose is a $[m, k^T, d^T]$ code, there are $k^2 = 16$ left logical qubits, $k = 4$ of which diagonal, and $(k^T)^2 = 1$ right logical qubits, $k^T = 1$ of which diagonal.
\begin{figure}
    \centering
    \includegraphics[scale=0.45]{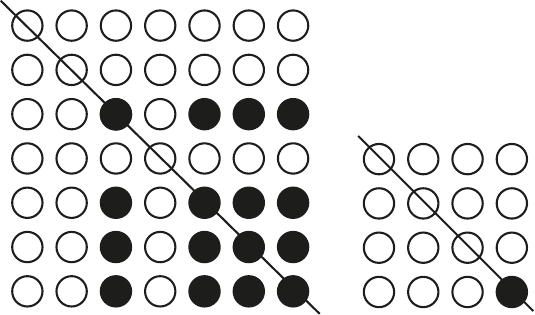}
    \caption{Graphical representation of the physical qubits of the square code $\hgp(\tilde H, \tilde H)$, where $\tilde H$ is as in \cref{eq:tilde_h}. The physical qubits (white circles) are arranged on a left and right grid. The black circles highlight the position of the physical pivots associated to the logical qubits, see \cref{tab:notation_qubits}. The canonical basis pictured is derived as explained in \Cref{lemma:basis} from the matrix $A$ of \cref{eq:a_matrix}, whose columns are a strongly lower triangular basis of $\ker \tilde H$.}
    \label{fig:logical_qubits_enumeration}
\end{figure}

\subsection{Gates on symmetric codes}
\label{sec:symmetric_transversal}
The second partition we propose improves on the diagonal-twin partition, having partition-distance $d$ against $\lfloor d/2 \rfloor$ for a $[\![n, k, d]\!]$ code. The price paid is the validity of this partition only on a small class of square codes, the symmetric codes: $\syhgp(H) = \hgp(H, H)$ for some $H$ such that $\im H = \im H^T$.

Notably, the toric code is a symmetric hypergraph product code $\syhgp(H_{\mathrm{toric}})$ where e.g.\ for $d = 3$,
\begin{align*}
    H_{\mathrm{toric}} = 
    \begin{pmatrix}
    1 & 1 & 0 \\
    0 & 1 & 1 \\
    1 & 0 & 1
    \end{pmatrix}.
\end{align*} 

For our purposes, the key feature of symmetric hypergraph product codes is that their physical qubits can be arranged in two square grids of the same size. This fact suggests that we can pair physical qubits by superimposing the left grid of physical qubits on the right one, and pairing qubits that sit at the same coordinate. In this way, every physical qubit $(i, h, L)$ is paired with its \emph{sibling} qubit $(i, h, R)$. 
Explicitly, given a symmetric code $\syhgp(H)$, with $H \in \ff^{n \times n}$, we define its sibling partition as:
\begin{align}
\label{sibling_partitioning}
\Big\{ \{(i, h, L),\, (i, h, R)\}\Big\}.
\end{align} 
where $1\le i, h, \le n$. See \cref{fig:symm_hamming} for a graphical representation of the sibling partition.
The sibling partition has partition-distance $\delta_s = d$. In fact, as said above and proven in~\cite{quintavalle2022reshape}, every non-trivial logical operator of a hypergraph product has support on at least $d$ rows or columns in the same sector. As such, even if the sibling partition is $2$-local, every subsets in it can give a contribution of at most $1$ towards the covering of an arbitrary logical operator. This observation is more general: any partition whose subsets $Q_i$ contain at most one qubit in each sector has maximum partition-distance $d$. We call any such partition sector-transversal.

The \gates{Hadamard-SWAP} operation defined above for the square codes on the diagonal-twin partition is similarly defined on the sibling partition too: physical Hadamard on all qubits and SWAP between siblings qubits yields Hadamard on all the logical qubits composed with logical SWAPs between sibling logical qubits.

Via the sibling partition is naturally defined a transversal \gates{CZ} operator. In fact, applying physical CZ gates on all pairs of sibling qubits preserves the $Z$ stabilisers and maps the $X$ stabilisers $S_{x}(j,h)$ into $S_{x}(j,h)S_{z}(h,j)$. On the logical level, the $\gates{CZ}$ operator yields a CZ between pairs of sibling logical qubits $q_{i, h}^L$ and $q_{i, h}^R$.
\subsubsection{An example of a symmetric code}
\label{sec:example_syhgp}
Building on our guide example in \Cref{sec:example_h_tilde}, we illustrate a symmetric hypergraph product code derived from the $[7, 4, 3]$ Hamming code with full-rank parity check matrix $H$,
\begin{align}\label{EqH_ind}
H=
\begin{pmatrix}
1 & 1 & 1 & 0 & 1 & 0 & 0\\
1 & 0 & 1 & 1 & 0 & 1 & 0\\
0 & 1 & 1 & 1 & 0 & 0 & 1
\end{pmatrix}.
\end{align}
We define the symmetric Hamming code as the symmetric hypergraph product code $\syhgp(H^TH)$, where\footnote{Note that, for any matrix $A$, $A^TA$ is symmetric.}
\begin{align*}
H^TH=\begin{pmatrix}
0 & 1 & 0 & 1 & 1 & 1 & 0\\
1 & 0 & 0 & 1 & 1 & 0 & 1\\
0 & 0 & 1 & 0 & 1 & 1 & 1\\
1 & 1 & 0 & 0 & 0 & 1 & 1\\
1 & 1 & 1 & 0 & 1 & 0 & 0\\
1 & 0 & 1 & 1 & 0 & 1 & 0\\
0 & 1 & 1 & 1 & 0 & 0 & 1
\end{pmatrix}.
\end{align*}
The symmetric Hamming code $\syhgp(H^TH)$ has parameters $[\![98,32,3]\!]$ and stabiliser generators of weight 8. 
It has a rate of $k/n=32/98\approx 0.33$ which substantially outperforms\footnote{For $d=3$, the surface and the toric code have rate $1/13\approx 0.08$ and $1/18 \approx 0.06$ respectively. 
The optimal rate for codes with $k=1$ is achieved by the five qubit code~\cite{Laflamme}.} the surface and toric code rates and indeed the optimal rate of $k/n=0.2$ for codes with $k=1$ and $d\geq 3$. For this reason, we believe that similarly constructed symmetric codes could be promising for low-overhead, near-term quantum computing. 
In \Cref{app:matrices} we present several more examples of small ($n<1000$) symmetric hypergraph product codes with stabiliser generators of weight $w\leq 16$ on which all the gates we describe could be implemented.

\begin{figure}
\centering
\begin{subfigure}[b]{0.48\columnwidth}
\centering
\includegraphics[scale=0.45]{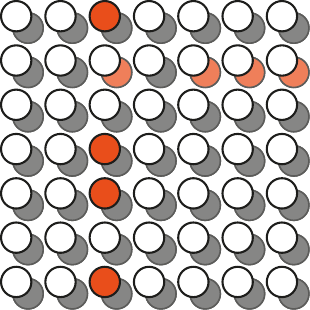}
\caption{
\label{fig:x_stab_hamming}
}
\end{subfigure}
\hfill
\begin{subfigure}[b]{0.48\columnwidth}
\centering
\includegraphics[scale=0.45]{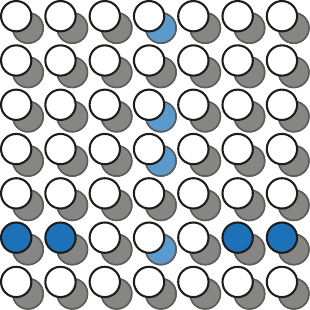}
\caption{
\label{fig:z_stab_hamming}
}
\end{subfigure}
\caption{Graphical representation of the symmetric hamming code $\syhgp(H^TH)$. In both figures, the left sector qubits sit at the front while the right sector qubits sit at the back. Sibling pair of qubits are superimposed. In \cref{fig:x_stab_hamming} the orange filled circles represent the support of the X stabiliser $S_x(2, 3)$.  In \cref{fig:z_stab_hamming} the blue filled circles represent the support of the Z stabiliser $S_z(4, 6)$. 
\label{fig:symm_hamming}
}
\end{figure}

\section{Completing a universal gate set \label{sec:universal}}

Transversal gates enable us to implement a subset of logical Clifford gates on hypergraph product codes, but this is not sufficient for performing universal fault-tolerant quantum computation. 
In particular, the Clifford gates detailed in \Cref{sec:transversal_gates} act equivalently on logical qubits of the same type (diagonal, twin or sibling) and so cannot be used to e.g.\ entangle arbitrary pairs of logical qubits. 
In this Section, we focus on symmetric hypergraph product codes and show how to implement a universal gate set using circuits that consist of sequences of sector-transversal gates\footnote{Transversal gates over a sector-transversal partition.}.
Such circuits are pieceably fault-tolerant, meaning that each gate in the sequence is fault-tolerant, and therefore the overall circuit is fault-tolerant when supplemented with intermediate error correction~\cite{Yoder}.
In \Cref{sec:pft}, we show how to implement entangling gates between arbitrary logical qubits by combining pieceably fault-tolerant gates. 
Then, in \Cref{sec:inj}, we show how to complete a universal gate via state injection. 
Lastly, in \Cref{sec:limit}, we discuss the limitations of our approach. 

\subsection{Pieceably fault-tolerant two-qubit entangling gates}
\label{sec:pft}
We begin by constructing a pieceably fault-tolerant circuit for implementing a logical CZ gate between logical qubits in different sectors, via the round-robin method presented in~\cite{Yoder}. 

By Claim 2 of~\cite{Chao}, we can implement a logical CZ between two logical qubits by performing a CZ between each pair of physical qubits in the support of the logical Z operators of the qubits considered. 
For a symmetric code (using the notation of \cref{tab:notation_qubits}) we can therefore implement a logical CZ between arbitrary qubits in different sectors, $q_{i, h}^L$ and $q_{j,\ell}^R$, by performing a CZ between each pair of physical qubits in the support of $\bar{Z}_{i,h}^L = (a_i \ox f_h, 0)$ and $\bar{Z}_{j, \ell}^R= (0, f_j \ox \beta_{\ell})$. 
For $v \in \ff^n$ we indicate by $\supp(v)$ the ordered set of qubits/indices in the support of the vector $v$ i.e. $\supp(v) = \{i \text{ s.t. } 1 \le i \le n \text{ and } v[i] = 1\}$. 
Claim 2 of~\cite{Chao} states that
\begin{align}\label{eq:chao}
&\prod_{(\eta, \gamma) \in \supp(a_i) \times \supp(\beta_{\ell})} \text{CZ } (\eta, h,L),(j,\gamma, R)
\end{align}
implements the logical operator
\begin{align*}
\text{CZ } q_{i,h}^L, q_{j,\ell}^R.
\end{align*}

In order to use the round-robin method to make the operation described by (\ref{eq:chao}) fault-tolerant, we want to group the physical CZ operations in separate fault-tolerant time steps and perform intermediate error correction between them. 
For symmetric codes, we can achieve this if we perform physical CZ's in tranches so that each tranche only contains CZ's between left and right sector qubits---as opposed to CZ's between two qubits from the same sector. 
To this end, we let $\Delta = \max(|\supp(a_i)|, |\supp(\beta_{\ell}|)$ and, for $\eta \in \supp(a_i)$, we denote by $\eta_{\#}$ its index in $\supp(a_i)$, meaning $\eta_{\#} = \nu$ if $\eta$ is the $\nu$th element in the ordered set $\supp(a_i)$; with a slight abuse of notation, we can write $\eta \oplus_{\Delta} t$ to indicate the $\mu$th element in $\supp(\beta_{\ell})$, where $\mu = \eta_{\#} + t \mod \Delta$, for any integer $t$.  Combining this notation with \cref{eq:chao}, yields
\begin{align}
\label{eq:pft_cz}
&\prod_{t=0}^{\Delta-1}\left(\prod_{\substack{\eta\in\supp(a_i) \\ \eta \oplus_\Delta t \leq |\supp(\beta_{\ell})|}} \text{CZ}(\eta, h,L),(j, \eta \oplus_{\Delta} t, R) \right). 
\end{align}
Treating each value of $t$ in \cref{eq:pft_cz} as a time step, the inner summations,
\begin{align}
\label{eq:internal}
\Omega_t = \prod_{\substack{\eta\in\supp(a_i) \\ \eta \oplus_\Delta t \leq |\supp(\beta_{\ell})|}} \text{CZ}(\eta, h,L),(j, \eta \oplus_\Delta t, R),
\end{align}
are sequences of non-overlapping sector-transversal operators, as desired. 
For each time step $0 \le t \le \Delta$, first we apply the gates of \cref{eq:internal} and then perform one round of error correction using the ReShape decoder~\cite{quintavalle2022reshape}. The use of ReShape at this stage is fundamental because it corrects errors on different sectors independently. Hence, using ReShape to correct errors between time steps, the protocol preserve the partition-distance $d$ of the operator $\Omega_t$ in  \cref{eq:internal} and can correct up to $\lfloor d/2 \rfloor$ faulty CZ gates.

In conclusion, CZ between arbitrary left and right logical qubits can be implemented in $\sim d_z^{\uparrow}$ time steps, where $d_z^{\uparrow}$ is the maximum weight of a canonical Z operator. An example on the symmetric Hamming code is depicted in \cref{fig:hamming_pft}.

\begin{figure}[]
    \centering
    \begin{subfigure}[b]{0.19\linewidth}
        \centering
        \includegraphics[scale=0.45]{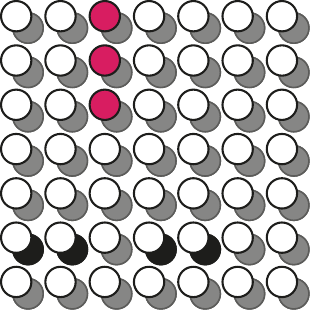}
        \caption{Logical $Z$
        \label{fig:hamming_logical}}
    \end{subfigure}
    \begin{subfigure}[b]{0.19\linewidth}
        \centering
        \includegraphics[scale=0.45]{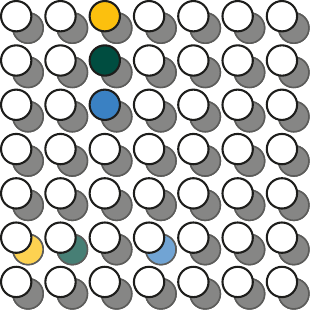}
        \caption{$t=0$.
        \label{fig:hamming_t0}
    }
    \end{subfigure}
    \begin{subfigure}[b]{0.19\linewidth}
        \centering
        \includegraphics[scale=0.45]{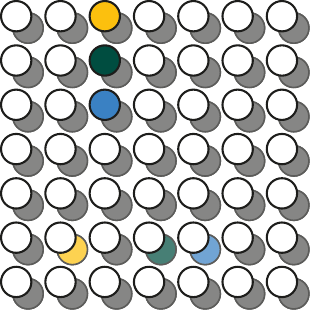}
        \caption{$t=1$.
        \label{fig:hamming_t1}}
    \end{subfigure}
    \begin{subfigure}[b]{0.19\linewidth}
        \centering
        \includegraphics[scale=0.45]{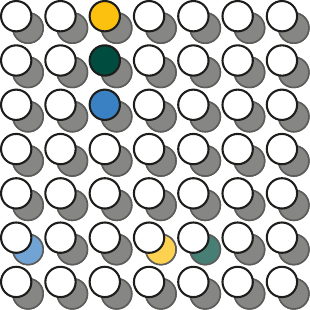}
        \caption{$t=2$.
        \label{fig:hamming_t2}}
    \end{subfigure}
    \begin{subfigure}[b]{0.19\linewidth}
        \centering
        \includegraphics[scale=0.45]{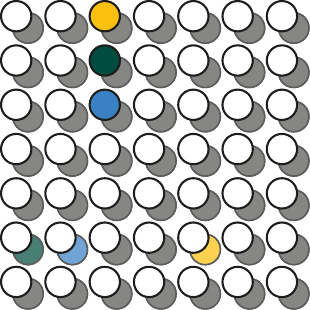}
        \caption{$t=3$.
        \label{fig:hamming_t3}}
    \end{subfigure}
    \caption{
        Round robin implementation of $\text{CZ } q_{3,3}^L q_{6,5}^R$ by the circuit of \cref{eq:pft_cz} on the symmetric Hamming code.
        (a) shows the supports of the logical Z operators. 
        (b)-(e) illustrate the physical CZ gates acting at each time step $t$, where qubits highlighted with the same colour have a CZ gate applied to them.
    }
    \label{fig:hamming_pft}
\end{figure}


We can also construct a pieceably fault-tolerant circuit for the gate 
\begin{equation}
\text{XCX} :=\text{H}^{\otimes 2}\cdot\text{CZ}\cdot \text{H}^{\otimes 2}
\end{equation}
Indeed, combining again Claim 2 of~\cite{Chao} and \cref{eq:pft_cz}, we find that
\begin{align}
\label{eq:pft2}
&\prod_{t=0}^{\Delta-1}\left(\prod_{\substack{\eta\in\supp(b_h)\\  \eta \oplus_\Delta t \leq |\supp(\alpha_j)|}} \text{XCX}(i, \eta, L),(\eta\oplus_{\Delta} t, \ell, R)\right) 
\end{align}
implements the logical operator
\begin{align*}
\text{XCX } q_{i,h}^L,q_{j,\ell}^R.
\end{align*}
where $\bar{X}_{i, h}^L = (f_i \otimes b_h, 0)$ and $\bar{X}_{j, \ell}^R =(0, \alpha_j \otimes f_{\ell})$ and $\Delta = \max(|\supp(b_h)|, |\supp(\alpha_j)|)$.
As in the CZ case, we can use the round-robin method in~\cite{Yoder} and perform error correction between each sector-transversal time step to obtain a fault-tolerant implementation of XCX.

Given our ability to perform CZ and XCX gates between arbitrary pairs of qubits in different sectors, we can use the circuit identity shown in \cref{fig:cnot_h} to implement CNOT gates between arbitrary pairs of qubits in the same sector.
To summarize, we can implement entangling gates between arbitrary logical qubits using circuits composed of at most four pieceably fault-tolerant circuits.
\begin{figure}[h]
\centering
\includegraphics[]{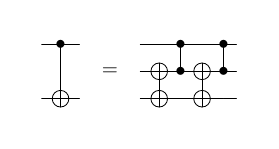}
\caption{
    Circuit identity showing how to construct CNOT from CZ and XCX gates, where XCX gates are depicted with targets on both qubits.
}
\label{fig:cnot_h}
\end{figure}

\subsection{Single-qubit gates via state injection}
\label{sec:inj}
To complete a universal gate set, it is sufficient to implement the H, S and T gates. The gates S and T can be implemented via state injection~\cite{gottesman1999,zhou2000,Bravyi}. We can use a CSS code as a state factory and then use a pieceably fault-tolerant CZ gate to implement the state injection circuits shown in \cref{fig:state_injection_gadget}. State injections can be performed from an auxiliary $[\![n', 1, d']\!]$ CSS code $\mathcal{C}$  that acts as a state factory by leveraging pieceably fault-tolerant gates to implement the circuits in \cref{fig:state_injection_gadget} at a logical level. 
More precisely, let $\zeta \in \ff^{n'}$ describe the support of the logical $Z$ operator for the logical qubit $q_{\mathcal C}$ of $\mathcal C$ and suppose that we number the physical qubits of $\mathcal C$ as $(\iota)_{\mathcal C}$.
Then, via Claim 2 of~\cite{Chao} and \cref{eq:pft_cz}, we obtain that
\begin{align}
\label{eq:cz_aux}
    &\prod_{t=0}^{\Delta-1}\left(\prod_{\substack{\eta\in\supp(a_i)\\ \eta \oplus_\Delta t \leq |\zeta|}} \text{CZ }(\eta, h,L), (\eta \oplus_{\Delta} t)_{\mathcal{C}}\right)
\end{align}
where $\Delta=\max\left(|\supp(a_i)|,|\supp(\zeta)|\right)$,
implements
\begin{align*}
    \text{CZ } q_{i,h}^L, q_{\mathcal{C}}.
\end{align*}
And similarly for a right qubit $q_{j, \ell}^R$.
The operator of \cref{eq:cz_aux} consists of CZ operators between different codes so is transversal at each time step and is therefore pieceably fault-tolerant.
As $\mathcal C$ is a CSS code, we can measure the logical $X$ operator destructively by measuring all of the physical qubits in the $X$ basis and performing classical error correction on the measurement outcomes. 
Therefore, to implement S, and T we only need to fault-tolerantly prepare the states $HS\ket{+}$ and $HT\ket{+}$.
For a $d=3$ symmetric hypergraph product code, one option would be to produce these states using the $[\![15,1,3]\!]$ Reed-Muller code~\cite{knill1996threshold,steane1999,anderson2014}.
For higher distance symmetric hypergraph product codes, preparing magic states could be accomplished using standard magic state distillation techniques~\cite{Bravyi,Bravyi2,Litinski,Chamberland}.
In addition, multiple magic states could also be injected from a CSS code with multiple encoded qubits, using parallel CZ gates (see \cref{sec:limit}).
\begin{figure}[t]
\centering
\includegraphics[]{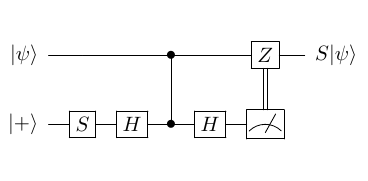} \hfill \includegraphics[]{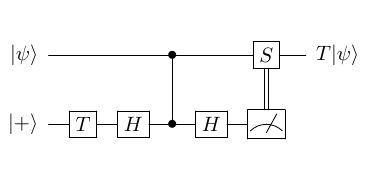} 
\caption{
    State injection gadgets for implementing S and T gates.
    The top wire in each circuit represents a logical qubit of a a symmetric hypergraph product code and the bottom wire represents a logical qubit of an ancillary CSS code. 
    To implement the CZ gates we use the pieceably fault-tolerant circuit in \cref{eq:pft_cz}.
    Note that for CSS codes the X basis measurements, represented as the gate H followed by Z basis measurements, can be done transversally.
}
\label{fig:state_injection_gadget}
\end{figure}

For the implementation of the logical H gate on single qubits we propose the use of the \texttt{Hadamard-Swap} operation detailed in \Cref{sec:symmetric_transversal} together with three logical S gates. Without loss of generality we here illustrate how to implement the gate H on the left logical qubit $q_{i, h}^L$ with sibling qubit $q_{i, h}^R$. The construction for right logical qubits follows easily by switching the roles of siblings. We indicate by $U(q)$ the logical gate $U$ on the logical qubit $q$, i.e.\ $\text{S}(q_{i, h}^L)$ indicates the logical S gate on the logical qubit $q_{i, h}^L$. The composition of gates:
\begin{align}
\label{eq:hadamard_single}
    \text{S}(q_{i,h}^L) \left( \texttt{ Hadamard-SWAP } \text{S}(q_{i, h}^R) \texttt{ Hadamard-SWAP }\right) \text{S}(q_{i, h}^L)
\end{align}
equates:
\begin{align*}
   \text{H S H S H }\left(q_{i,h}^L \right).
\end{align*}
Since
\begin{align}
    \text{H } (\text{S H S}) \text{ H} = e^{i \pi/4} \text{ H} 
\end{align}
the composition of gates in \cref{eq:hadamard_single} implements the logical H gate on qubit $q_{i, h}^L$, up to a global phase factor. Hence, the logical H gate on arbitrary qubits can be implemented at the cost of three S gates and two sector-transversal operations. 
\subsection{Limitations of pieceable fault-tolerance}
\label{sec:limit}

We conclude this Section by highlighting two important limitations of pieceable fault tolerant techniques: intermediate correction and time cost.
To guide our analysis, let us consider the pieceably fault-tolerant CZ gate of \Cref{sec:pft}. 

Since CZ is diagonal in the $Z$ basis, $Z$ stabilisers are left unchanged during the protocol and hence correction for $X$ errors can be done, at each time step, via standard measurement of $Z$ stabilisers. 
However, the original $X$ stabilizer generators are transformed at each intermediate time step, with no guarantee on their weight---which could scale with the code distance.
To avoid measuring high-weight generators, we can neglect to do error correction for $Z$ errors at intermediate time steps, at the cost of allowing $Z$ errors to build up for a constant number of rounds (for codes of a fixed distance).
This strategy will only be fault-tolerant for codes of a fixed distance and therefore will not give a threshold without modification\footnote{Another possible solution is to perform weight reduction~\cite{Cohen,Hastings,Hastings2} on the $X$ stabiliser generators at each time step, likely at the cost of a high space overhead.}. 
Nevertheless, the asymptotic behaviour encapsulated by a threshold is less important in the short- to medium-term regime where space overhead will likely be the most important constraint.

As regards to time cost, pieceably fault-tolerant circuits necessarily have a time overhead higher than transversal gates because of intermediate error-correction. 
If the time cost of one cycle of error-correction followed by the application of a transversal gate is $\tau$, then each sector-transversal gate has time cost $\tau$. By contrast, the pieceably fault-tolerant CZ of \cref{eq:pft_cz}, has cost at least $\tau d$ for a distance $d$ code. For instance, in the symmetric Hamming code example given in \cref{sec:example_syhgp}, CZ or XCX have cost $\tau 4$ because $4$ is the maximum weight of a canonical logical operator.

In general, composing $m$ pieceable fault tolerant gates takes time $\tau  d^{\uparrow} m$, where $d^{\uparrow}$ is the maximum weight of a canonical logical operator. 
However, this overhead can be reduced via parallelization. For instance,  
\begin{align*}
    \text{CZ } q_{i, h}^L, \, q_{j, \ell}^R
    \quad \text{and} \quad
    \text{CZ } q_{i', h'}^L, \, q_{j', \ell'}^R
\end{align*}
can be performed in parallel via \cref{eq:pft_cz}, provided that the logical Z operators of the qubits considered have all disjoint support e.g.\ whenever $h \neq h'$ and $j \neq j'$; see \cref{fig:parallel} for an example.
Therefore, if we want to implement $m$ pieceable fault-tolerant gates, we can divide them into subsets of parallelizable gates and, if $n_p$ is the minimum number of gates in any of these subsets, we can implement all the gates in time $\tau d^\uparrow m/n_p$, 
reducing the average time overhead per gate. 
Importantly, performing gates in parallel has no additional space overhead, in contrast to the measurement-based scheme of~\cite{Cohen}.

\begin{figure}
    \centering
    \includegraphics[scale=0.45]{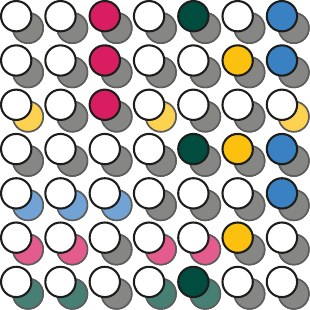}
    \caption{
        Graphical representation of the parallel implementation of pieceably fault tolerant gates. 
        Columns of left (front) qubits and rows of right qubits (back) filled in the same color represent pairs of logical operators whose supports do not overlap. 
        In the image, magenta: $q_{3, 3}^L, q_{6, 5}^R$, green $q_{7, 5}^L, q_{7, 5}^R$, yellow: $q_{6, 6}^L, q^R_{3, 7}$, blue: $q_{5, 3}^L q_{5, 7}^R$. 
        Pieceably fault tolerant gates that act on the support of these logical operators (e.g.\ CZ) can be implemented in parallel.  
        We note that, for any symmetric hypergraph product code of dimension $k^2$, we can always find $k$ pairs of logical qubits whose logical operators do not overlap as shown here (one for each left line pivot and one for each right line pivot).
    }
\label{fig:parallel}
\end{figure}

\section{Conclusion}
\label{sec:discussion}

We have investigated the structure of hypergraph product codes to characterize transveral Clifford gates on them. In fact, we proved that every hypergraph product code has a canonical basis for its logical space whose features suggest qualitative differences between logical qubits (diagonal, twin and sibling). We showcased how to leverage such differences to implement some transversal Clifford gates on square and symmetric hypergraph product codes. The logical transversal gates we have found are limited but we comment on how these could be augmented via pieceable fault tolerance and state injection to achieve universality.

While we have focused on symmetric hypergraph product codes, we believe that the choice of structured seed classical codes in the hypergraph product could be exploited to construct other transversal gates. We also conjecture that a partitioning strategy similar to the ones here presented could be used to implement transversal non-Clifford gates on higher dimensional homological codes~\cite{vasmer2019three, jochym2021four}. More generally, it would be interesting to extend the construction of a canonical basis to the more efficient product codes construction such as the ones in \cite{Panteleev2, Hastings3, Breuckmann2, leverrier2022quantum}.

 \clearpage
\acknowledgements
This work is supported by the Australian Research
Council via the Centre of Excellence in Engineered Quantum Systems (EQUS) project number CE170100009.
Research at Perimeter Institute is supported in part by the Government of Canada through the Department of Innovation, Science and Economic Development Canada and by the Province of Ontario through the Ministry of Colleges and Universities. We acknowledge support by the Open Access Publication Fund of Freie Universität Berlin.\newline\noindent
The authors would like to thank Stephen Bartlett, Earl Campbell and Lawrence Cohen for helpful discussions.
\bibliographystyle{quantum.bst}
\bibliography{bibliography_doi.bib}

\appendix
\section{Proofs}
\label{app:proof}
In this Section we introduce the strongly lower triangular Gaussian reduction for binary matrices (\Cref{algo:strong_triangular}) and prove its correctness. 

\Cref{algo:strong_triangular} takes in input a binary matrix $H$ and outputs a strongly lower triangular basis for its kernel, as per \Cref{def:strong_triang}, and a unit vector basis for its complement \footnote{We remind the reader that, given a vector space $V \subseteq \ff^n$, its complement $V^{\bullet}$ is any vector space such that $V \oplus V^{\bullet} = \ff^n$, see the proof of \Cref{lemma:basis}.}. It iterates over the columns of $H$ (line \ref{for:first}) and iteratively removes from the set of columns indices (line \ref{if:pivots}) the independent columns; the matrix of interest, $K$, which will contain the column vectors forming a basis of the kernel of $H$ is iteratively reduced in triangular form (line \ref{for:second}). We prove that \Cref{algo:strong_triangular} is correct in \Cref{lemma:correctness} below.
\begin{algorithm}[H]
\caption{Strong lower triangular Gaussian reduction.}
\begin{algorithmic}[1]
\Require An $m \times n$ binary matrix $H$.
\newline
\Ensure A strongly triangular basis of its kernel $\mathcal K$ and a unit vector basis $\mathcal F$ of $(\im H^T)^{\bullet}$, where $\ff^n = \im H^T \oplus (\im H^T)^{\bullet}$. 
\newline
\State{$\hat H\gets H$}
\State{$K \gets I_n$}
\State{$\pi(\mathcal{K}) = \{1, \dots, n\}$}
\For{$j\gets 1$ \textbf{to} $n$} \label{for:first}
    \While {$\hat H[i][j] \neq 1$ and $1 \le i \le m$}
    \State{$i \gets i+1$}
    \EndWhile
    \If{$\hat H[i][j] = 1$}: \label{if:pivots}
        \State{$\pi(\mathcal{K}) \gets \pi(\mathcal{K}) \setminus \{j\}$}
            \For{$\ell \gets j+ 1$ \textbf{to} $n$} \label{for:second}
            \If{$\hat H[i][\ell] = 1$} \label{if:ker}
                \State{$\hat H^\ell \gets \hat H^j + \hat H^\ell$}
                \State{$K^{\ell} \gets K^j + K^{\ell}$} \label{state:ker}
                \EndIf
                \EndFor
                \EndIf
                \EndFor
\State{$\mathcal{K} \gets \{K^j : j \in \pi(\mathcal{K})\} $}
\State{$F \gets \{I_n^j : j \in \pi(\mathcal{K})\}$}
\State\Return{$\mathcal{K}$, $\mathcal{F}$}
\end{algorithmic}
\label{algo:strong_triangular}
\end{algorithm}
\begin{lemma}
\label{lemma:correctness} 
\Cref{algo:strong_triangular} terminates and is correct. More precisely:
\begin{enumerate}[1)]
    \item The set $\mathcal{H} = \{H^j : 1\le j \le n, \, j \not\in \pi(\mathcal{K})\}$ is a basis of the column span of $H$.
    \item The set $\mathcal{K}$ is a basis of the kernel of $H$.
    \item The set $\mathcal{F}$ is a unit-vector basis of $(\im H^T)^{\bullet}$. 
\end{enumerate}
\end{lemma}
\begin{proof}
Let $v$ be any vector in the column span of $H$:
\begin{align*}
    v = \sum_{j \in V} H^j.
\end{align*}
If $V \cap \pi(\mathcal K) = \emptyset $ then there is nothing to prove. Conversely, suppose that there exists $j \in V$ such that $j \in  \pi(\mathcal{K})$. Then it must be $\hat H^j = 0$ (see line \ref{if:pivots}). In other words, the $j$th column can be written as a linear combination of other columns of $H$, and by substitution if necessary, point 1) is proved. 

To prove point 2), observe that for what said on the zero columns of $\hat H$, all the vectors in $\mathcal{K}$ are in the kernel of $H$. Moreover, they are linearly independent and lower triangular by construction. Now suppose that the matrix is not strongly lower triangular (e.g.\ $\begin{pmatrix}
    1 & 1\\
    0 & 1
\end{pmatrix}$) so that there exists a row index $i$ such that $K_{i, \alpha} = K_{i, \alpha+\tau} = 1$ for some $\tau > 0$. This means that the $\alpha$th column $K^{\alpha}$ has been added to the $(\alpha+\tau)$th column $K^{\alpha+\tau}$ at step $\alpha$. However, if $\alpha \in \pi(\mathcal{K})$ is a pivot index, $\hat H^\alpha$ is zero at steps $\alpha -1, \dots, n$ and therefore the if condition at line \ref{if:ker} is not met.

Point 3) follows observing that the vectors in $\mathcal{F}$ are linearly independent by construction (they are unit vectors with different pivots) and none of them belongs to the row span of $H$. In fact by definition, any vector in the kernel of $H$ is orthogonal to vectors in the row span of $H$, and by strong lower triangularity, $\langle k, f\rangle  = 1$ for any $k \in \mathcal{K}$ and $f \in \mathcal{F}$. 
\end{proof}
\section{Symmetric hypergraph product codes - Examples}
\label{app:matrices}
We here provide \cref{tab:matrices} with some examples of classical seed matrices to construct symmetric hypergraph product codes with $n<1000$ and maximum stabiliser weight $w\leq 16$. The classical parity check matrices  used in the product were found with assistance from the ``Best Known Linear Code'' database of MAGMA~\cite{MAGMA} and guidance from the bounds in~\cite{Grassl}.
For each parity check matrix $H$ in the first column, we have computed the $[\![n, k, d]\!]$ parameters the symmetric hypergraph product code $\syhgp(H^TH)$.

\begin{table}[h!]
  \begin{tabular}{| c | c | c | c | c | c |}
    \hline
   {Classical Code Parity Check}  & {$n$} & {$k$} & {$k/n$} & {$d$} & {$w$} \\ \hline
\tiny{$\begin{matrix} &&&&&& \\ 1&1&1&0&1&0&0\\1&0&1&1&0&1&0\\0&1&1&1&0&0&1 \\ &&&&&& \end{matrix}$} & 98 & 32 & 0.33 & 3 & 8\\ \hline
\tiny{$\begin{matrix}       &   &  &   &  & & &  &  &   &\\ 0  &   0   &  0  &   1   &  1  &   1  &   0   &  1   &  0  &   0 &    0\\
     0   &  1 &    1 &    1   &  1   &  0  &   1 &    0&     1&     0   &  0\\
     1 &    0  &   1  &   1   &  0  &   1  &   1   &  0  &   0   &  1  &  0\\
     1  &   1   &  0  &   0&    1 &    1 &    1 &    0  &   0 &    0  &   1 \\  &   &  &   &  & & &  &  &   & \end{matrix}$} & 242 & 98 & 0.41 & 3 &12 \\ \hline
\tiny{$\begin{matrix}    &&&&&&  &&&&&&&&\\        0  &   0   &  0  &   1 &    1 &    1&     1  &   1   &  1 &    1  &   0  &   1 &    0 &    0  &   0\\
     0 &    1  &   1  &   0  &   0  &   1 &    1   &  1  &   1 &    0   &  1  &   0&     1  &   0 &    0\\
     1  &   0  &   1  &   0  &   1  &   0  &   1   &  1 &    0   &  1 &    1 &    0  &   0  &   1&     0\\
     1  &   1   &  0   &  1  &   0    & 0     &1     &0     &1&     1   &  1  &   0   &  0   &  0 &    1 \\ &&&&&&  &&&&&&&& \end{matrix}$} & 450 & 242 & 0.54 & 3 &16 \\ \hline
\tiny{$\begin{matrix} &&&&&& \\ 1&1&0&1&0&0&0\\
1&0&1&0&1&0&0\\
0&1&1&0&0&1&0\\
1&1&1&0&0&0&1\\
&&&&&& \end{matrix}$} & 98 & 18 & 0.18 & 4 &8\\ \hline
\tiny{$\begin{matrix} &&&&&&&&&&&\\
1&1&1&0&0&0&0&1&0&0&0&0\\
0&0&0&1&1&1&0&0&1&0&0&0\\
1&1&0&1&1&0&1&0&0&1&0&0\\
1&0&1&1&0&1&1&0&0&0&1&0\\
0&1&1&0&1&1&1&0&0&0&0&1 \\
&&&&&&&&&&&
\end{matrix}$} & 288 & 98 & 0.34 & 4 &12\\ \hline
\tiny{$\begin{matrix}       
   &  &   &  & &   &   &  & &   \\
 1   &  1 &    0 &  1&     0   &  0   &  0 &    0    & 0  &   0\\
     1 &    0   &  1   &  0   &  1     &0   &  0&    0&     0 &    0\\
     0  &   1 &    1  &   0 &    0   &  1 &    0  &   0   &  0  &   0\\
     1  &   1  &   1 &    0 &    0  &   0   &  1  &   0 &    0&     0\\
     1  &   0 &    0   &  0 &    0&     0  &   0   &  1 &   0   &  0\\
     0  &   1 &    0 &    0 &    0&    0    & 0 &    0 &    1  &   0\\
     0  &   0   &  1  &   0   &  0   &  0 &    0   &  0  &   0   &  1 \\
 &  &   &  & &   &   &  & &   
     \end{matrix}$} & 200 & 18 & 0.09 & 5 &8 \\ \hline
\tiny{$\begin{matrix}      
&&	&&&&	&	&		&	&\\
 0	&1&	1	&0	&1	&0&	0&	0&	0	&0	&0\\
1&	0&	1&	0&	0&	1&	0&	0&	0&	0&	0\\
1&	1&	1&	0&	0&	0&	1&	0&	0&	0&	0\\
1&	1&	0&	1&	0&	0&	0&	1&	0&	0&	0\\
1&	0&	1&	1&	0&	0&	0&	0&	1&	0&	0\\
0&	1&	1&	1&	0&	0&	0&	0&	0&	1&	0\\
1&	1&	1&	1&	0&	0&	0&	0&	0&	0&	1 \\
&&	&&&&	&	&		&	&
\end{matrix}$} & 242 & 32 & 0.13 & 5 &16 \\ \hline
\tiny{$\begin{matrix}      
&&&&&&&&&&&&&\\
1&1&1&0&1&0&0&0&0&0&0&0&0&0\\
1&1&0&1&0&1&0&0&0&0&0&0&0&0\\
1&0&1&1&0&0&1&0&0&0&0&0&0&0\\
0&1&1&1&0&0&0&1&0&0&0&0&0&0\\
0&0&1&1&0&0&0&0&1&0&0&0&0&0\\
1&1&1&1&0&0&0&0&0&1&0&0&0&0\\
1&0&0&1&0&0&0&0&0&0&1&0&0&0\\
0&1&0&1&0&0&0&0&0&0&0&1&0&0\\
1&0&1&0&0&0&0&0&0&0&0&0&1&0\\
0&1&1&0&0&0&0&0&0&0&0&0&0&1\\
&&&&&&&&&&&&&
 \end{matrix}$} & 392 & 32 & 0.08 & 7 &16 \\ \hline
\tiny{
$\begin{matrix}      
&&&&&&&&&&&&&&&&&&\\
1&0&0&0&1&0&0&0&0&0&0&0&0&0&0&0&0&0&0\\
0&1&0&0&0&1&0&0&0&0&0&0&0&0&0&0&0&0&0\\
0&0&1&0&0&0&1&0&0&0&0&0&0&0&0&0&0&0&0\\
0&0&0&1&0&0&0&1&0&0&0&0&0&0&0&0&0&0&0\\
1&1&0&0&0&0&0&0&1&0&0&0&0&0&0&0&0&0&0\\
1&0&1&0&0&0&0&0&0&1&0&0&0&0&0&0&0&0&0\\
1&0&0&1&0&0&0&0&0&0&1&0&0&0&0&0&0&0&0\\
0&1&0&1&0&0&0&0&0&0&0&1&0&0&0&0&0&0&0\\
0&1&1&0&0&0&0&0&0&0&0&0&1&0&0&0&0&0&0\\
0&0&1&1&0&0&0&0&0&0&0&0&0&1&0&0&0&0&0\\
1&1&1&0&0&0&0&0&0&0&0&0&0&0&1&0&0&0&0\\
1&1&0&1&0&0&0&0&0&0&0&0&0&0&0&1&0&0&0\\
1&0&1&1&0&0&0&0&0&0&0&0&0&0&0&0&1&0&0\\
0&1&1&1&0&0&0&0&0&0&0&0&0&0&0&0&0&1&0\\
1&1&1&1&0&0&0&0&0&0&0&0&0&0&0&0&0&0&1 \\
&&&&&&&&&&&&&&&&&&
\end{matrix}$} & 722 & 32 & 0.04 & 9 &16 \\ \hline
  \end{tabular}
\caption{Examples of symmetric hypergraph product codes $\syhgp(H^TH)$, where $H$ is the binary matrix in the first column. The length $n$, the dimension $k$, the rate $k/n$, the distance $d$ and the maximum stabiliser weight $w$ are reported in the other columns.
\label{tab:matrices}}
\end{table}

\end{document}